\documentclass[twoside,11pt]{article}
\usepackage[utf8]{inputenc}		
\usepackage[margin=1in]{geometry}
\usepackage{jeffe,graphicx}
\usepackage{microtype}
\usepackage{cite}
\usepackage{scalerel}

\usepackage[charter]{mathdesign}
\usepackage{sourcesanspro,inconsolata,eucal}

\newtheorem{lemma}{Lemma}[section]
\newtheorem{theorem}[lemma]{Theorem}
\newtheorem{corollary}[lemma]{Corollary}

\def\Twin{\operatorname{\emph{twin}}}
\def\Sign{\operatorname{\emph{sign}}}
\def\Edge{\operatorname{\emph{edge}}}
\def\Cross{\operatorname{\emph{cross}}}
\def\Param{\vec}

\usepackage{arydshln}
\begin{document}

\title{Optimal Curve Straightening is $\exists\Real$-Complete}
\author{Jeff Erickson \\
	\href{http://jeffe.cs.illinois.edu}{University of Illinois, Urbana-Champaign}}
\date{}	
\maketitle

\begin{abstract}
We prove that the following problem has the same computational complexity as the existential theory of the reals: Given a generic self-intersecting closed curve $\gamma$ in the plane and an integer $m$, is there a polygon with $m$ vertices that is isotopic to $\gamma$?  Our reduction implies implies two stronger results, as corollaries of similar results for pseudoline arrangements.  First, there are isotopy classes in which every $m$-gon with integer coordinates requires $2^{\Omega(m)}$ bits of precision.  Second, for any semi-algebraic set $V$, there is an integer $m$ and a closed curve $\gamma$ such that the space of all $m$-gons isotopic to $\gamma$ is homotopy equivalent to~$V$.   
\end{abstract}

\pagestyle{myheadings}
\markboth{Optimally Curve Straightening is $\exists\Real$-Complete}{Jeff Erickson}

\section{Introduction}

The image of any self-intersecting closed curve $\gamma$ in the plane has a natural structure as a (not necessarily simple) $4$-regular plane graph, whose vertices are the points where $\gamma$ intersects itself.  After subdividing if necessary to remove any loops or parallel edges, the classical Wagner-Fáry-Stein-Stojaković theorem implies that there is a combinatorially equivalent plane map whose edges are straight line segments.  Thus, any curve with $n$ self-intersections can be continuously deformed, without ever changing its intersection pattern, into a polygon with at most $O(n)$ vertices.  This type of continuous deformation is called an \emph{isotopy}.

While this $O(n)$ bound is tight in the worst case, it is far from optimal for all curves.  For example, straightening the curve in Figure~\ref{F:example} by finding a piecewise-linear embedding of its image graph yields the decagon on the left, but in fact the curve is isotopic to the self-intersecting hexagon on the right.  As a more extreme example, consider any polygon in which \emph{every} pair of edges intersects, for example, the regular star polygon with Schäfli symbol $\set{m, \floor{m/2}}$ for some odd integer $m$.  Interpreting such a polygon as a curve and then converting that curve back to a polygon by straightening its image graph roughly \emph{squares} the number of vertices.

\begin{figure}[ht]
\centering
\includegraphics[scale=0.5]{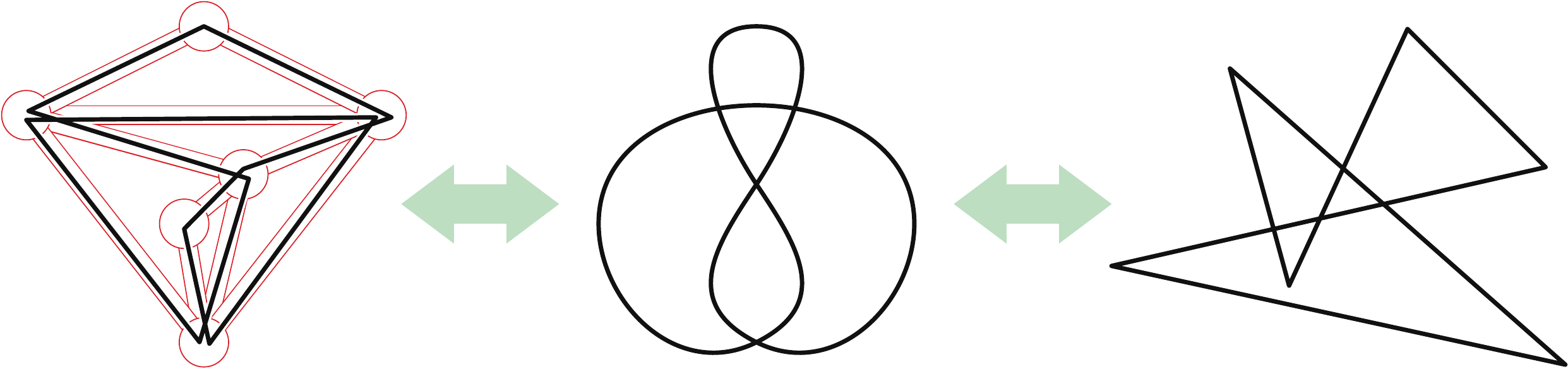}
\caption{Two polygons isotopic to the same closed curve.}
\label{F:example}
\end{figure}

This quadratic gap between best and worst cases motivates the following question: How quickly can we find a polygon with the minimum number of vertices that is isotopic to a given closed curve?  This is a special case of \emph{minimum-segment drawings} problem proposed by Dujmović \etal \cite{desw-dpgws-07}, which asks for the minimum number of line segments whose union is isomorphic, as a straight-line planar graph, to a given plane graph.

The main result of this paper is that deciding whether a given closed curve is isotopic to an \emph{arbitrary} polygon with a given number of vertices has the same computational complexity as \emph{the existential theory of the reals}---the problem of deciding whether a given system of multivariate polynomial equalities and inequalities has a real solution.  In standard computational complexity terms, our result implies that optimal curve-straightening is both NP-hard and in \textsc{PSpace}.

In Section \ref{S:hard}, we prove that curve-straightening is hard using a simple reduction from the stretchability problem for pseudoline arrangements, which was proved equivalent to the existential theory of the reals by Mnëv \cite{m-utcpc-88,m-vctpc-85}.  Figure \ref{F:reduction} shows a typical example of our reduction.  Our argument is similar to, but both simplifies and strengthens, a proof by Durocher \etal \cite{dmnw-nmdpg-13} that computing minimum-segment drawings of plane graphs is $\exists\Real$-hard, which relies on a reduction by Bose \etal~from pseudoline stretchability to recognizing arrangement graphs~\cite{bew-pag-03}. 

\begin{figure}[ht]
\centering\footnotesize\sffamily
\begin{tabular}{c@{\qquad\qquad}c}
\raisebox{-0.5\height}{\includegraphics[scale=0.4]{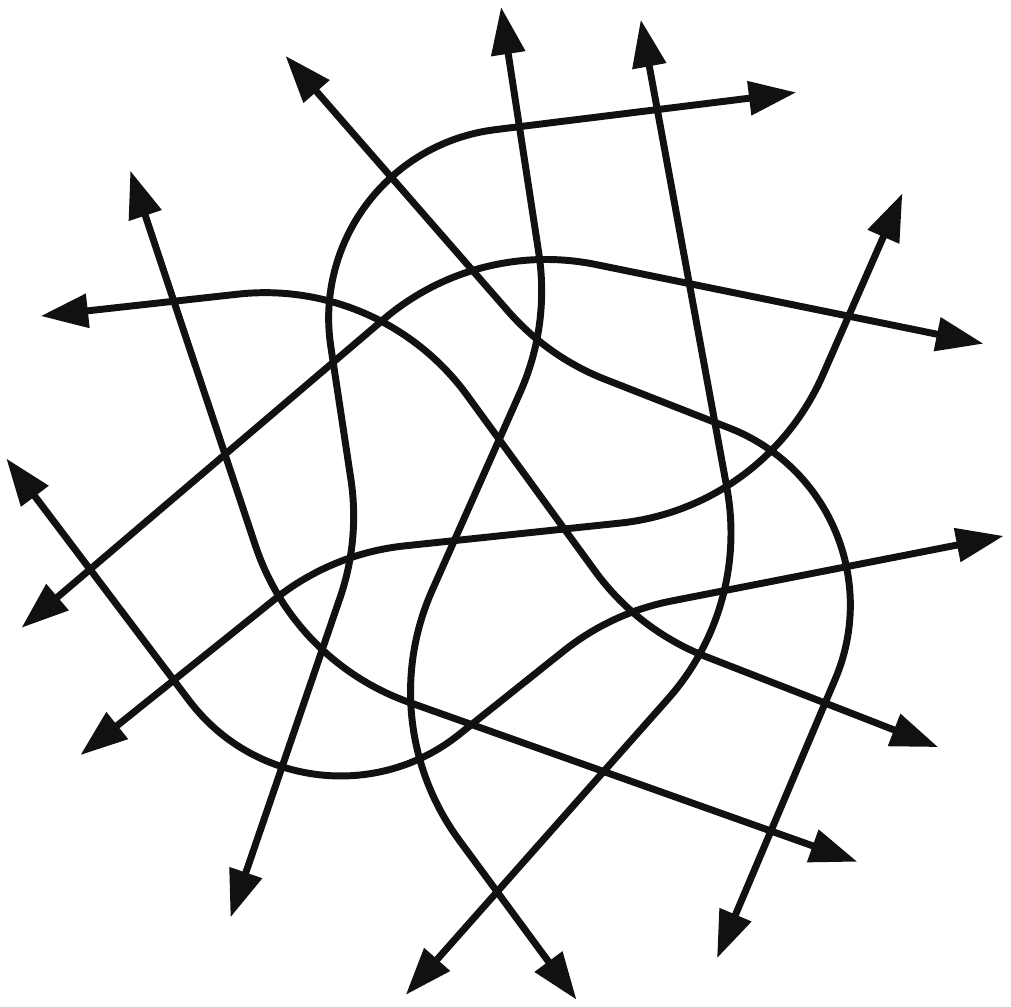}}&
\raisebox{-0.5\height}{\includegraphics[scale=0.4]{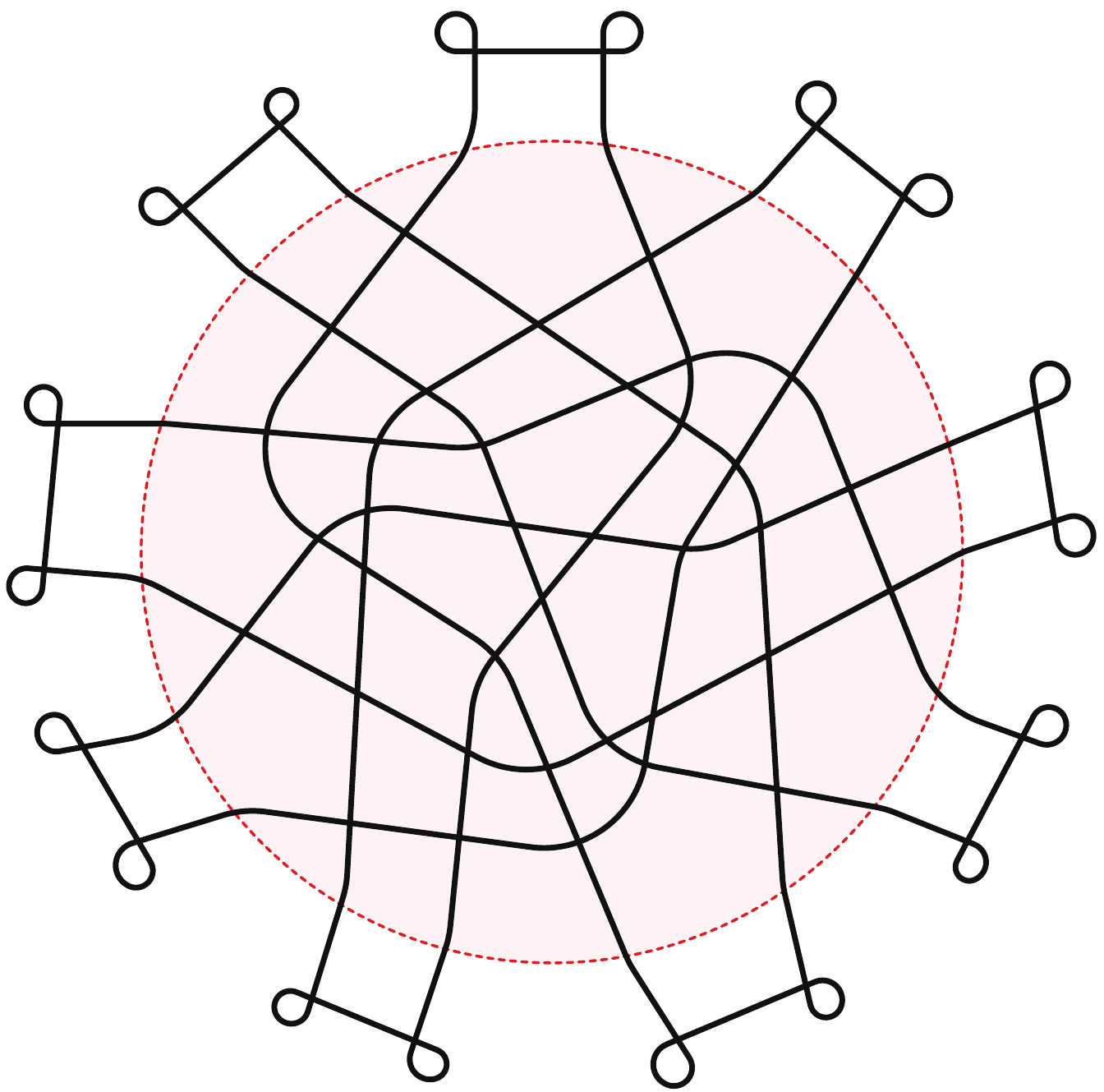}}\\(a)&(b)
\end{tabular}
\caption{(a) Ringel's \cite{r-tedgo-56} minimal non-stretchable arrangement of 9 pseudolines.  (b) A closed curve derived from Ringel's arrangement that is not isotopic to any 36-gon.}
\label{F:reduction}
\end{figure}

Together with Mnëv's universality theorem for pseudoline arrangements \cite{m-utcpc-88,m-vctpc-85}, our reduction implies a similar universality result for polygons: the set of all polygons isotopic to a given curve and with a  given number of vertices can have essentially arbitrary topology.  Similarly, results of Goodman \etal~\cite{gps-iscm-90,gps-crotre-89} imply that an integer polygon isotopic to a given curve could require doubly-exponentially large coordinates in the worst case.

Our reduction \emph{to} the existential theory of the reals in Section \ref{S:easy} is surprisingly more delicate than our hardness proof, for two reasons.  First, we must be much more explicit about how the isotopy class of the input curve is represented; we use a variant of codes proposed by Gauss \cite{g-n1gs-00,g-nglzr-00} and Dowker and Thistlethwaite \cite{dt-ckp-83}, with some additional assumptions to establish an embedding in the plane instead of the sphere.  More significantly, there is no direct correspondence between the vertex coordinates of the polygon and the conditions that guarantee isotopy.  (By contrast, for example, reducing the recognition of segment intersection graphs to ETR requires exactly one clause for each pair of vertices, stating either that the corresponding segments cross or that they do not \cite{km-igs-94,m-igse-14}.)  We exploit the nondeterminism of ETR by guessing which edges of the isotopic polygon contain each self-crossing; using this information essentially requires us to implement integer arrays inside our formula.

Interestingly, given a closed curve $\gamma$ with $n$ self-intersection points, we can compute the minimum-complexity \emph{orthogonal} polygon isotopic to $\gamma$ in $O(n^{4/3}\polylog n)$ time, by combining Tamassia's reduction from minimum-bend grid embeddings to minimum-cost flows \cite{t-eggwm-87} with a recent planar minimum-cost flow algorithm of Karczmarz and Sankowski~\cite{ks-mfupg-19}; see also Cornelson and Karrenbauer \cite{ck-abm-12}.

\section{Background}
\label{S:background}

\subsection{Closed Curves}

A \emph{parametrized closed curve} is a continuous function $\Param\gamma\colon S^1\to \Real^2$ from the circle $S^1 = \Real/\Z = [0,1]/\set{0,1}$ to the plane.  In this paper, we consider only \emph{generic} closed curves, which are injective except at a finite number of transverse double points.  The \emph{basepoint} of a parametrized closed curve $\Param\gamma$ is the point $\Param\gamma(0)$.

Two parametrized closed curves $\Param\gamma$ and $\Param\gamma'$ are \emph{equivalent} if they differ only by a reparametrization, meaning there is a homeomorphism $\theta\colon S^1 \to S^1$ such that $\Param\gamma' = \Param\gamma\circ \theta$.  Equivalent parametrized curves have exactly the same images, but possibly different basepoints and orientations.  A \emph{closed curve} is an equivalence class of parametrized closed curves under this equivalence relation.


An \emph{ambient isotopy} of the plane is a continuous family $H\colon [0,1]\to(\Real^2 \to \Real^2)$ of homeomorphisms of the plane, where without loss of generality $H(0)$ is the identity map.  Two parametrized closed curves~$\Param\gamma$ and $\Param\gamma'$ are \EMPH{isotopic} if there is an ambient isotopy $H$ such that $\Param\gamma' = H(1)\circ\Param\gamma$, and two closed curves are isotopic is they have isotopic parametrizations.  Less formally, two curves are isotopic if there is a continuous deformation of the plane taking (the image of) one to (the image of) the other.

The \emph{image graph} of a closed curve $\gamma$ is the plane graph whose vertices are the curve's self-intersection points, and whose edges are maximal subpaths of the curve between vertices.  Two closed curves are isotopic if and only if their image maps are combinatorially isomorphic (forbidding reflection).  Thus, up to isotopy, closed curves can be represented either by any reasonable data structure for topological planar maps \cite{t-cpm-63,w-ebdss-85,c-age-97,gs-pmgsc-85, mp-fitcp-78}, or by more specialized structures encoding the order and directions of self-intersections along the curve \cite{g-n1gs-00,g-nglzr-00,t-k1-1876,c-cic-91,dt-ckp-83}.  We describe one such encoding in Section~\ref{SS:codes}.

We regard any polygon $P$ with vertices $(p_1, p_2, \dots, p_m)$ as a piecewise-linear parametrized closed curve with basepoint $p_1$ that visits the vertices $p_i$ in increasing index order.  A polygon is \emph{generic} if all vertices are distinct, no vertex lies in the interior of an edge, and no three edges contain a common point.

\begin{lemma}
\label{L:crude-upper}
Every generic closed curve with $n$ self-crossings is isotopic to a generic polygon with $O(n)$ vertices.
\end{lemma}

\begin{proof}
Let $\gamma$ be a generic closed curve with $n$ self-crossings.  Subdivide the edges of the image graph of~$\gamma$ as necessary to to obtain a simple plane graph $G$.  Let $\bar{G}$ be a straight-line plane graph that is combinatorially equivalent (and therefore isotopic) to $G$.  Let $P$ be the unique Euler tour of $\bar{G}$ that crosses itself at every vertex; $P$ is a polygon with (crudely) at most $6n$ vertices.  Finally, for sufficiently small~$\e$, uniformly and independently perturbing each vertex of $P$ inside a ball of radius $\e$ yields an isotopic polygon $\tilde{P}$, which is generic with probability $1$. 
\end{proof}

\subsection{Pseudoline Arrangements}

An \EMPH{pseudoline arrangement} is a cellular decomposition of the (projective) plane by a finite set of infinite curves (called \emph{pseudolines}), each of which is isotopic to a straight line, and any pair of which intersect transversely, exactly once.  In particular, every arrangement of (pairwise non-parallel) straight lines is a pseudoline arrangement.  A pseudoline arrangement is \EMPH{simple} if each point in the plane lies on at most two pseudolines.  Two pseudoline arrangements are \EMPH{isomorphic} if they define combinatorially isomorphic decompositions, or equivalently \cite{blswz-om-00} if there is a homeomorphism of the plane to itself that maps one arrangement to the other.

A pseudoline arrangement is \EMPH{stretchable} if it is isomorphic to an arrangement of (pairwise non-parallel) lines.  The first non-stretchable pseudoline arrangements were constructed by Levi \cite{l-tpedg-26}, based on classical theorems of projective geometry; Ringel later refined Levi's construction to obtain the first \emph{simple} non-stretchable pseudoline arrangements~\cite{r-tedgo-56}.  Figure~\ref{F:reduction}(a) shows one of Ringel's arrangements, based on Pappus's hexagon theorem.\footnote{This is in fact the smallest simple non-stretchable pseudoline arrangement.  Every arrangement of at most eight pseudolines is stretchable \cite{gp-pgcsc-80}, and up to isomorphism, there is exactly one simple non-stretchable arrangement of nine pseudolines \cite{r-krom-89}.}

For further background on pseudoline arrangements and stretchability, see Felsner and Goodman \cite{fg-pa-17}, Björner \etal~\cite{blswz-om-00}, Bokowski and Sturmfels \cite{bs-csg-89}, or Grünbaum \cite{g-as-72}.

\subsection{The Existential Theory of the Reals}

The \emph{existential theory of the reals} (ETR) is the set of all true existentially quantified boolean combinations of multivariate polynomial inequalities.  More formally, we can recursively define a (fully parenthesized) \emph{sentence} in the existential theory of the reals (in \emph{prenex form}) as follows:
\begin{itemize}
\item
A \emph{polynomial} is either the constant $0$, the constant $1$, a variable $x_i$, or an expression of the form $(Q+Q')$ or $(Q-Q')$ or $(Q \cdot Q')$ for some polynomials $Q$ and $Q'$.
\item
An \emph{atomic predicate} is a comparison of the form $(Q=0)$ or $(Q<0)$ for some polynomial $Q$.
\item
A \emph{predicate} is either an atomic predicate or an expression of the form $(\lnot \Phi)$ or $(\Phi\lor \Phi')$ or $(\Phi\land \Phi')$ for some predicates $\Phi$ and $\Phi'$.  The set of all real vectors $(x_1, x_2, \dots, x_N)\in\Real^N$ that satisfy such a predicate is called a \emph{semi-algebraic set}.
\item
Finally, a \emph{sentence} is an existentially quantified formula $\exists x_1, x_2, \dots, x_N \colon \Phi(x_1, x_2, \dots, x_N)$, where $\Phi(x_1, x_2, \dots, x_N)$ is any predicate in which the only variables are $x_1, x_2, \dots, x_N$.
\end{itemize}
ETR is the set of all such sentences that are true when all variables are quantified over the set of real numbers.  Equivalently, ETR is the problem of deciding whether a given predicate describes a non-empty semi-algebraic set.

The precise computational complexity of ETR is unknown.  The problem is easily seen to be NP-hard (see \cite[p.~513]{bpr-arag-06}, for example), and Canny proved that $\text{ETR}\in\textsc{PSpace}$ \cite{c-sagcp-88}.  The fastest algorithm known to solve ETR, due to Renegar \cite{r-ccgft1-92}, runs in time $(nd)^{O(N)}$, where $n$ is the number of atomic predicates in the input sentence,~$d$~is the maximum degree of any polynomial, and $N$ is the number of distinct real variables; see also Basu~\etal~\cite{bpr-arag-06}.

The definition of ETR is actually quite flexible.  We can freely include arbitrary bounded-degree algebraic real constants, subtraction, division, square roots, absolute values, more general comparisons ({$\le$, $\ne$, $>$,~$\ge$}), and more general boolean operators ($\Rightarrow$, $\Leftrightarrow$, $\oplus$, etc.); these can all be eliminated with only a small increase in the length of the sentence.  On the other hand, ETR remains $\exists\Real$-complete (see below) if we allow only strict inequalities, or only a single equation $\exists\vec{x}: Q(\vec{x}) = 0$, or only conjunctions of atomic predicates of the form $x_i<x_j$, $x_i+x_j=x_k$, and $x_i\cdot x_j = x_k$ \cite{s-spn-91}.

\subsection{$\exists\Real$-hardness and Universality}

Schaefer \cite{s-csgtp-09, s-rgl-13} and Schaefer and Štefankovič \cite{ss-fpnee-17} defined the complexity class $\exists\Real$ as the set of all decision problems that can be (many-one) reduced to ETR in polynomial time.  A problem is \emph{$\exists\Real$-hard} if there is a polynomial-time (many-one) reduction from ETR to that problem, and a decision problem is \emph{$\exists\Real$-complete} if it is both in $\exists\Real$ and $\exists\Real$-hard.\footnote{$\exists\Real$ is also equal to the Blum-Shub-Smale complexity class $\set{0,1}^*\cap \text{NP}^0_{\Real}$, the set of boolean vectors recognizable in polynomial time by a non-deterministic BSS machine over the reals \cite{bcss-crc-98} whose only constants are $0$ and $1$~\cite{bc-cccnc2-06,s-csgtp-09}.}

A seminal result of Mnëv \cite{m-utcpc-88,m-vctpc-85} implies that deciding whether a given simple pseudoline arrangement is stretchable is $\exists\Real$-complete; indeed, this was the first problem (other than ETR itself) to be identified as $\exists\Real$-complete.  In fact, Mnëv proved a significantly stronger \emph{universality} result.  The \emph{realization space} of an arrangement of $n$ pseudolines is the set of all isomorphic line arrangements, each encoded as a vector in~$\Real^{2n}$.  Mnëv proved that for any (open) semialgebraic set~$V$, there is a (simple) pseudoline arrangement whose realization space is homotopy equivalent to $V$.  Mnëv's proof relies on von Staudt's \emph{algebra of throws} \cite{vs-bg-1847}, which implements arithmetic using projective geometric constructions.  More detailed proofs of Mnëv's results were later given by Shor \cite{s-spn-91} and Richter-Gebert \cite{r-mutr-95,r-rsp-96}.

Mnëv's universality theorem has since been used to prove the $\exists\Real$-hardness of  several other geometric and topological problems.  Examples include
realizability of abstract 4-polytopes \cite{r-rsp-96},
cross-product systems \cite{hsz-scptc-13},
Delaunay triangulations~\cite{pt-dtwdr-14,apt-utipd-15},
and planar linkages \cite{km-utcspl-02,s-rgl-13};
recognition of several types of graphs \cite{bew-pag-03,km-igs-94,m-igse-14,mm-irdsg-13,s-rgl-13,ch-rcpvg-17,cfmtv-igrgs-18};
several problems in graph drawing \cite{dmnw-nmdpg-13,h-cpsnp-17,cflrvw-cdgfl-17,ck-csgge-15,cfmtv-igrgs-18,b-sphcn-91,s-csgtp-09};
several problems related to fixed points and Nash equilibria~\cite{ss-fpnee-17,bm-edpas-17,bm-cedpa-16,gmvy-edvms-18};
problems in social choice theory \cite{p-rmep-17};
optimal polytope nesting \cite{dhm-utnp-19,s-nrmhp-17};
nonnegative matrix factorization \cite{cr-nrdfn-93,dhm-utnp-19,s-nrmhp-17};
minimum rank of a matrix with a given sign pattern \cite{bfgjk-vcfo-09,bk-ccmrs-15};
real tensor rank \cite{ss-ctr-18};
and the art gallery problem \cite{aam-agpc-18}.
(Some of these papers claim only NP-hardness but describe reductions from pseudoline stretchability that imply $\exists\Real$-hardness; others make no claims about computational complexity but derive appropriate universality theorems that imply $\exists\Real$-hardness.\footnote{Mnëv's universality theorem also implies several other universality results in algebraic geometry, collectively dubbed “Murphy's Law” by Vakim \cite{v-magbb-06,lv-mus-13}.})

For more examples and further background on $\exists\Real$-hardness, we refer the reader to surveys by Schaefer \cite{s-csgtp-09}, Matoušek \cite{m-igse-14}, and Cardinal \cite{c-cgc62-15}.

\section{Straightening Curves is Hard}
\label{S:hard}

We consider the decision problem \textsc{CurveToPolygon}, which asks, given a nonnegative integer $m$ and (any reasonable representation of) a generic closed curve $\gamma$ in the plane, whether any polygon with~$m$ vertices is isotopic to $\gamma$.

To prove that \textsc{CurveToPolygon} is $\exists\Real$-hard, we describe a straightforward polynomial-time reduction from deciding stretchability of simple pseudoline arrangements.  Given a simple arrangement of $n$ pseudolines $\Psi = \set{\psi_1, \psi_2, \dots, \psi_n}$, we construct a corresponding generic closed curve $\gamma_\Psi$ that is isotopic to a $4n$-gon if and only if $\Psi$ is stretchable.  Figure \ref{F:example} gives a pictorial summary of our reduction.

We assume without loss of generality that the curves in $\Psi$ constitute a \emph{wiring diagram}, as defined by Goodman~\cite{g-pcbgs-80}.  That is, the curves are piecewise-linear, and each curve is horizontal except in a small neighborhood of each crossing.  We assume the pseudolines are indexed from top to bottom to the left of all crossings, and therefore from bottom to top to the right of all crossings, as shown in Figure \ref{F:pseudolines}.

\begin{figure}[ht]
\centering
\includegraphics[scale=0.5]{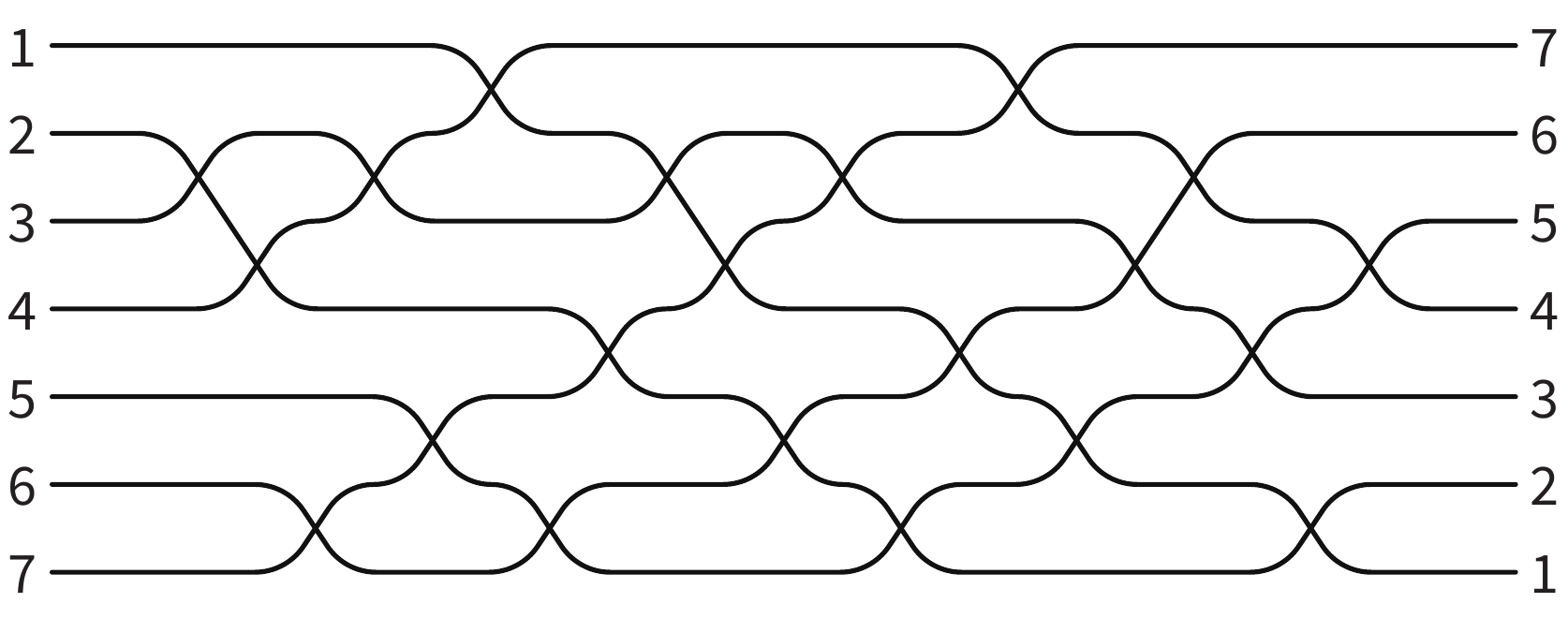}
\caption{A “wiring diagram” pseudoline arrangement.}
\label{F:pseudolines}
\end{figure}

To avoid parity issues, we also assume without loss of generality that $n$ is odd.  Otherwise, consider the set $\Psi' = \Psi \cup \set{\psi_0}$, where $\psi_0$ is a new “wire” that passes above all $\binom{n}{2}$ intersection points of $\Psi$ and then downward across all $n$ wires in $\Psi$ to the right of those intersection points.  We easily observe that the arrangement of $\Psi'$ is simple, and that $\Psi'$ is stretchable if and only if $\Psi$ is stretchable.

We derive our closed curve $\gamma_\Psi$ from the wiring diagram $\Psi$ as follows.  Truncate the pseudolines to an axis-aligned rectangle $R$ that contains all crossings in its interior.  For each index~$i$, let~$u_i$ and~$v_i$ denote the left and right endpoints of the curves $\Psi_i\cap R$.  We connect the truncated pseudolines using~$n$ pairwise-disjoint paths outside $R$, each containing two simple outward-facing loops and no other self-intersections, and each connecting two adjacent endpoints on the boundary of $R$, as shown in Figure~\ref{F:braidcurve}.  Specifically, for each index $i$, we connect the left endpoints $u_{2i}$ and~$u_{2i+1}$ to the left of $R$, and we connect the right endpoints $v_{2i-1}$ and~$v_{2i}$ to the right of $R$; finally, we connect the top endpoints $u_1$ and $v_{2n}$ above~$R$.  Concatenating these $n$ exterior paths and the $n$ truncated pseudolines yields the closed curve~$\gamma_\Psi$.  We collectively refer to the $2n$ exterior loops as the \emph{fringe} of $\gamma_\Psi$.

\begin{figure}[ht]
\centering
\includegraphics[scale=0.5]{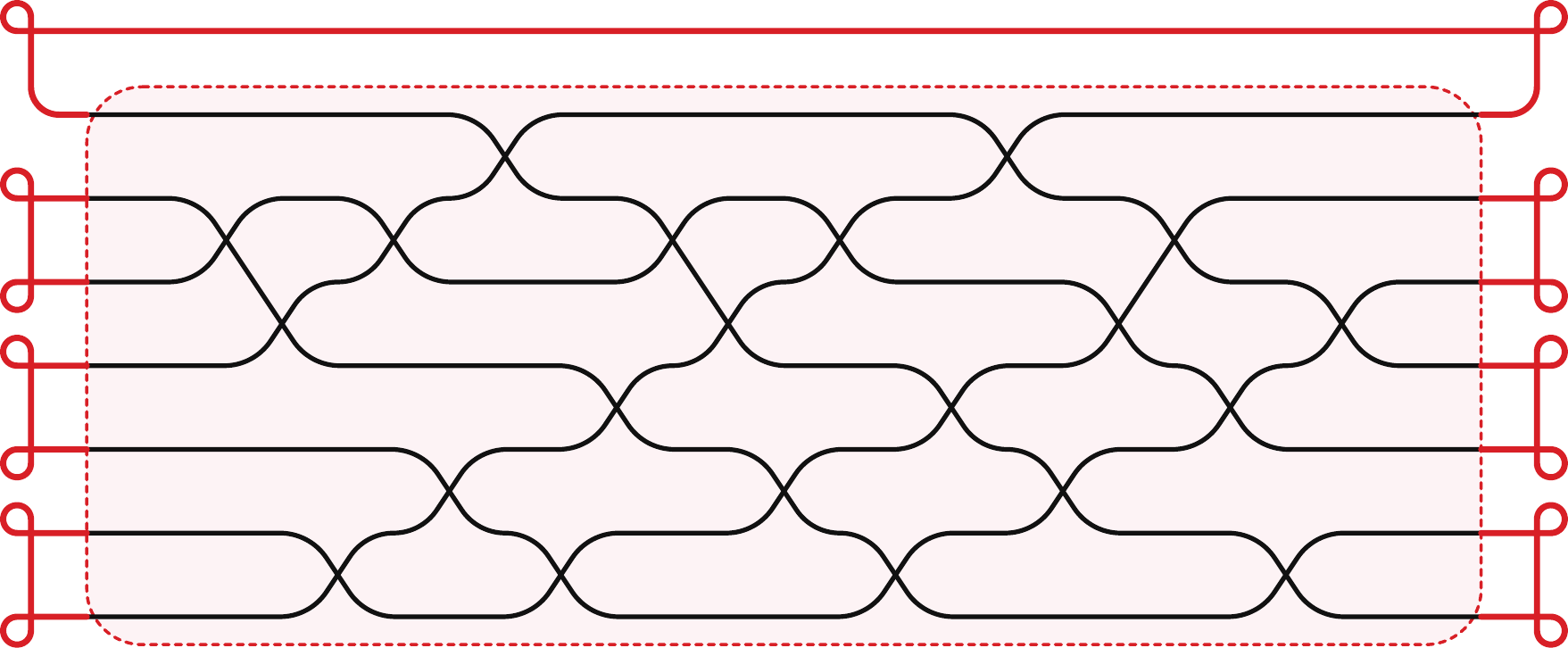}
\caption{The closed curve $\gamma_\Psi$ derived from the pseudoline arrangement $\Psi$ in Figure \ref{F:pseudolines}.}
\label{F:braidcurve}
\end{figure}

\begin{lemma}
\label{L:fringe}
Let $\Psi$ be any set of $n$ pseudolines, for any odd integer $n$, whose arrangement is simple.  Any polygon isotopic to the closed curve $\gamma_\Psi$ has at least $4n$ vertices.
\end{lemma}

\begin{proof}
Let $P$ be any polygon isotopic to $\gamma_\Psi$, and let $H\colon \Real^2\to\Real^2$ be any homeomorphism such that $P = H\circ \gamma_\Psi$.  Consider any simple loop $\ell$ in the fringe of $\gamma_\Psi$.  The corresponding path $L = H\circ\ell$ is a simple loop of $P$, which starts and ends at the unique intersection point of two edges $e$ and $e'$ of $P$.  These two edges are distinct, and because they cross, they do not share an endpoint.  Thus, $L$ contains at least two vertices of $P$, specifically, one endpoint of $e$ and one endpoint of $e'$.

\begin{figure}[ht]
\centering
\includegraphics[scale=0.5]{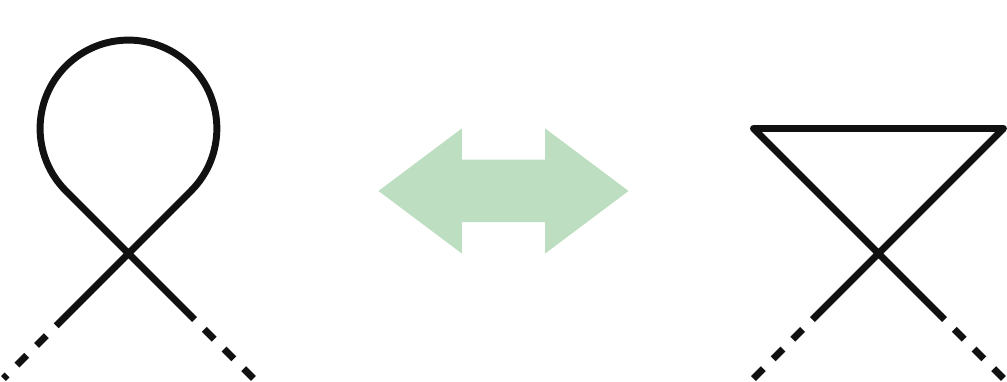}
\caption{Polygonalizing a simple loop requires at least two vertices}
\end{figure}

It follows immediately that $P$ contains at least two vertices for each of the $2n$ loops in the fringe of~$\gamma_\Psi$, and thus at least $4n$ vertices in total.  
\end{proof}

\begin{lemma}
\label{L:reduction}
Let $\Psi$ be any set of $n$ pseudolines, for any odd integer $n$, whose arrangement is simple.  The curve $\gamma_\Psi$ is isotopic to a $4n$-gon if and only if the arrangement of $\Psi$ is stretchable.
\end{lemma}

\begin{proof}
First, suppose $P$ is a $4n$-gon isotopic to $\gamma_\Psi$.  Recall that $R$ is an axis-aligned rectangle containing all the vertices of the arrangement of $\Psi$ in its interior.  Let $H\colon \Real^2\to\Real^2$ be any homeomorphism such that $P = H\circ \gamma_\Psi$, and let $D = H(R)$ be the image of $R$ under this homeomorphism.  The proof of Lemma~\ref{L:fringe} implies that $D$ does not contain any vertex of $P$.  Thus, the intersection $P\cap D$ is an arrangement of $n$ line segments, each pair of which crosses exactly once.  By construction, the homeomorphism $H$ maps the arrangement of arcs $\Psi \cap B^2$ to the arrangement of segments $P\cap D$.  Extending these segments to lines does not introduce any new intersection points or change the order of intersections along any curve.  Thus, the resulting line arrangement is isomorphic to the arrangement of $\Psi$; in short, $\Psi$ is stretchable.

\medskip
On the other hand, suppose $L$ is a set of lines in the plane whose arrangement is isomorphic to the arrangement of $\Psi$.  Let $H\colon \Real^2\to\Real^2$ be any homeomorphism that maps the arrangement of $\Psi$ to the arrangement of $L$, and for each index $i$, let $l_i$ denote the line $H\circ \psi_i$.  Without loss of generality, the lines in $L$ are indexed by increasing slope, each line in $L$ has slope strictly between $-1$ and $1$, each pair of lines in $L$ intersect strictly between the vertical lines $x=1$ and $x=1$.

For each index $i$, let $p_i$ and $q_i$ denote the intersection of $l_i$ with the lines $x=-5/4$ and $x=5/4$.  We can construct a $4n$-gon $P_L$ isotopic to $\gamma_\Psi$ by connecting the segments $p_iq_i$ with gadgets we call \emph{staples}, each consisting of a segment of length $\e$ and slope $1$, followed by a horizontal or vertical segment, followed by a segment of length $\e$ and slope $-1$, for some suitable real $\e>0$ to be chosen later.  Specifically, we connect $p_1$ to $q_n$ with a horizontal staple, and for each index $i$, we connect the left endpoints $p_{2i-1}$ and $p_i$ and the right endpoints $q_{2i}$ and $q_{2i+1}$ with vertical staples, as shown in Figure \ref{F:staples}.  Each staple crosses the two lines that it connects, creating two simple loops.  For sufficiently small $\e$, the staples are pairwise disjoint, and so the resulting $4n$-gon $P_L$ is isotopic to $\gamma_\Psi$.
\end{proof}

\begin{figure}[ht]
\centering
\includegraphics[scale=0.5]{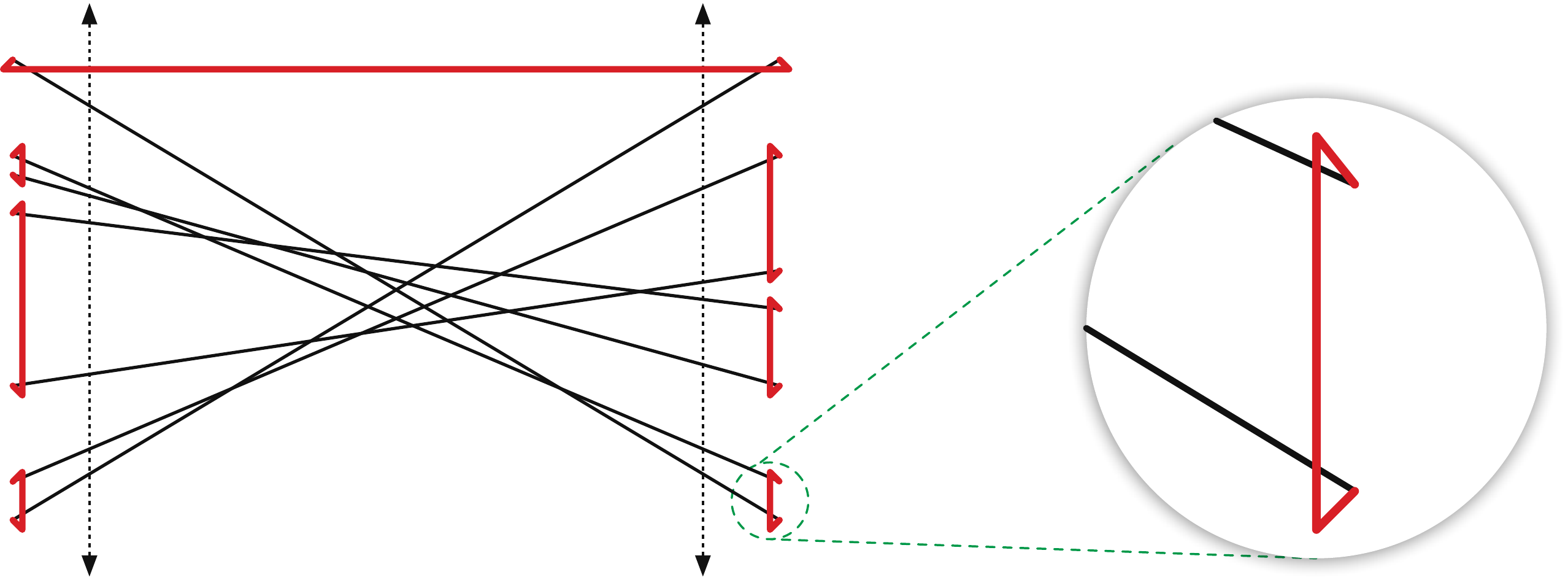}
\caption{Stapling a line arrangement into a polygon; compare with Figure \ref{F:braidcurve}.}
\label{F:staples}
\end{figure}

Given any set $\Psi$ of pseudolines, we can easily construct the curve $\gamma_\Psi$ in polynomial time.  Thus, Lemma \ref{L:reduction} immediately implies the following:

\begin{theorem}
\textsc{CurveToPolygon} is $\exists\Real$-hard.
\end{theorem}

In fact, our reduction implies much stronger result, as a natural corollary of Mnëv's universality theorem for pseudoline arrangements \cite{m-utcpc-88,m-vctpc-85,s-spn-91}.  The \emph{coordinate vector} of a polygon $P$ is a vector $(x_1, y_1, x_2, y_2, \dots, x_m, y_m)\in \Real^{2m}$, where $(x_1,y_1)$, $(x_2,y_2)$, \dots, $(x_m, y_m)$ are the coordinates of $P$'s vertices in cyclic order.  Two polygons (or their coordinate vectors) are \emph{equivalent} if they are equivalent as closed curves, or equivalently, if they differ only by a rotation or reflection of vertex indices.  Thus, every generic $m$-gon is equivalent to exactly $2m-1$ other $m$-gons.  Finally, for any closed curve $\gamma$ and any integer~$m$, we define the \emph{isotopy realization space} $\Sigma(\gamma,m) \subseteq \Real^{2m}/D_m$ to be the set of all equivalence classes of coordinate vectors of $m$-gons that are isotopic to $\gamma$.  (Here $D_m$ is the dihedral group of permutations acting on the $m$ vertex indices.)

\begin{theorem}
Every semialgebraic set is homotopy equivalent to the isotopy realization space $\Sigma(\gamma,m)$ of some closed curve $\gamma$ and some integer $m$.
Moreover, every \textbf{open} semialgebraic set is homotopy equivalent to the isotopy realization space $\Sigma(\gamma,m)$ of some \textbf{generic} closed curve $\gamma$ and some integer $m$.
\end{theorem}

Every generic $m$-gon is isotopic to a generic $m$-gon with integer coordinates.  In light of our reduction, results of Goodman, Pollack, and Sturmfels \cite{gps-iscm-90,gps-crotre-89} on the intrinsic spread of order types imply that some isotopy classes of polygons require an exponential number of bits to describe.  (McDiarmid and Müller \cite{mm-irdsg-13} derive a similar lower bound for intersection graphs of segments of unit disks in the plane.)

\begin{corollary}
For every integer $m$, there is an $m$-gon $P$ such that every integer $m$-gon isotopic to $P$ has diameter $2^{2^{\Omega(m)}}$.
\end{corollary}

\section{\dots But Not Too Hard}
\label{S:easy}

Reducing our curve-straightening problem to ETR is somewhat more subtle.  Intuitively, the reduction appears straightforward, because there is a simple efficient algorithm to verify that an $n$-crossing curve $\gamma$ and an $m$-vertex polygon~$P$ are isotopic.  First, compute all self-intersection points of $P$, and their order along each edge of $P$.  If $P$ has either more or less than $n$ self-intersections, or if any self-intersection of $P$ is degenerate (at a vertex, or on more than two edges), fail immediately.  Otherwise, fix an arbitrary basepoint $\gamma(0)$ on $\gamma$, and then for each vertex $p_i$ of $P$, check whether corresponding crossings appear in the same order and orientation along both curves, starting at $p_i$ on $P$ and starting at $\gamma(0)$ on $\gamma$.

But an efficient verification algorithm is insufficient; we require an efficient formula in the existential theory of the reals.  That is, we need to argue that the realization space of all $m$-gons isotopic to a given curve $\gamma$ is (the projection of) a semialgebraic set with  description complexity polynomial in~$n$. 

%

\subsection{Crossing Codes}
\label{SS:codes}

Our reduction assumes that the closed curve $\gamma$ is represented using a variant of encodings proposed by Gauss \cite{g-n1gs-00,g-nglzr-00}, Tait \cite{t-k1-1876}, and Dowker and Thistlethwaite \cite{dt-ckp-83}.\footnote{Specifically, Gauss \cite{g-n1gs-00,g-nglzr-00} encoded curves by assigning each crossing point a unique label, and then listing the labels in order along the curve, either with or without the corresponding sequence of signs.  Tait \cite{t-k1-1876} observed that each label in any Gauss code appears once at an even position and once in an odd position, the labels at (say) even positions can be discarded with no loss of information.  Dowker and Thistlethwaite observed that the crossing involution we are calling $\Twin$ is completely determined by the subsequence $(\Twin_1, \Twin_3, \dots, \Twin_{2n-1})$, because  $\Twin_i$ is even if and only if $i$ is odd \cite{dt-ckp-83}.}  Fix an arbitrary parametrization~$\Param\gamma$ of $\gamma$, such that the basepoint $\gamma(0)$ is not a self-intersection point.  Call $t\in S^1$ a \emph{crossing value} for $\Param\gamma$ if $\Param\gamma(t) = \Param\gamma(t')$ for some $t'\ne t$ (which must also be a crossing value).  Let $0 < t_1 < t_2 < \cdots < t_{2n} < 1$ denote the sorted sequence of crossing values.  The \EMPH{signed crossing code} of $\Param\gamma$ consists of a pair of vectors:
\begin{itemize}
\item
The first vector \textbf{$\Twin$}${}\in [2n]^{2n}$ is an involution of $[2n]$ that describes the pairing of crossing values; $\Param\gamma(t_i) = \Param\gamma(t_j)$ and $\Twin_j = i$ if and only if $j = \Twin_i$.
\item
The second vector \textbf{$\Sign$}${}\in \set{-1, +1}^{2n}$ indicates the direction of each crossing; $\Sign_i = +1$ if and only if, for arbitrarily small $\e$, the subpath $\Param\gamma[t_i-\e, t_i+\e]$ crosses the subpath $\Param\gamma[t_{\Twin_i}-\e, t_{\Twin_i}+\e]$ from right to left.
\end{itemize}
Each (unparametrized) closed curve $\gamma$ has up to $4n$ different signed crossing codes, depending on the choice of basepoint and direction.

A signed crossing code of a curve determines its embedding \emph{on the sphere} \cite{c-cic-91}, but not necessarily \emph{in the plane}, because the code does not specify an outer face.  To remove this ambiguity, we declare that the outer face lies immediately to the right of the curve near the basepoint.  Thus, to encode a planar curve, we must first choose a basepoint on the outer face and a direction that puts the outer face to the right of that basepoint.  Figure \ref{F:codes} shows a closed curve whose image graph has three arcs on the outer face, each leading to a different signed crossing code.

\begin{figure}[ht]
\centering
\includegraphics[scale=0.5]{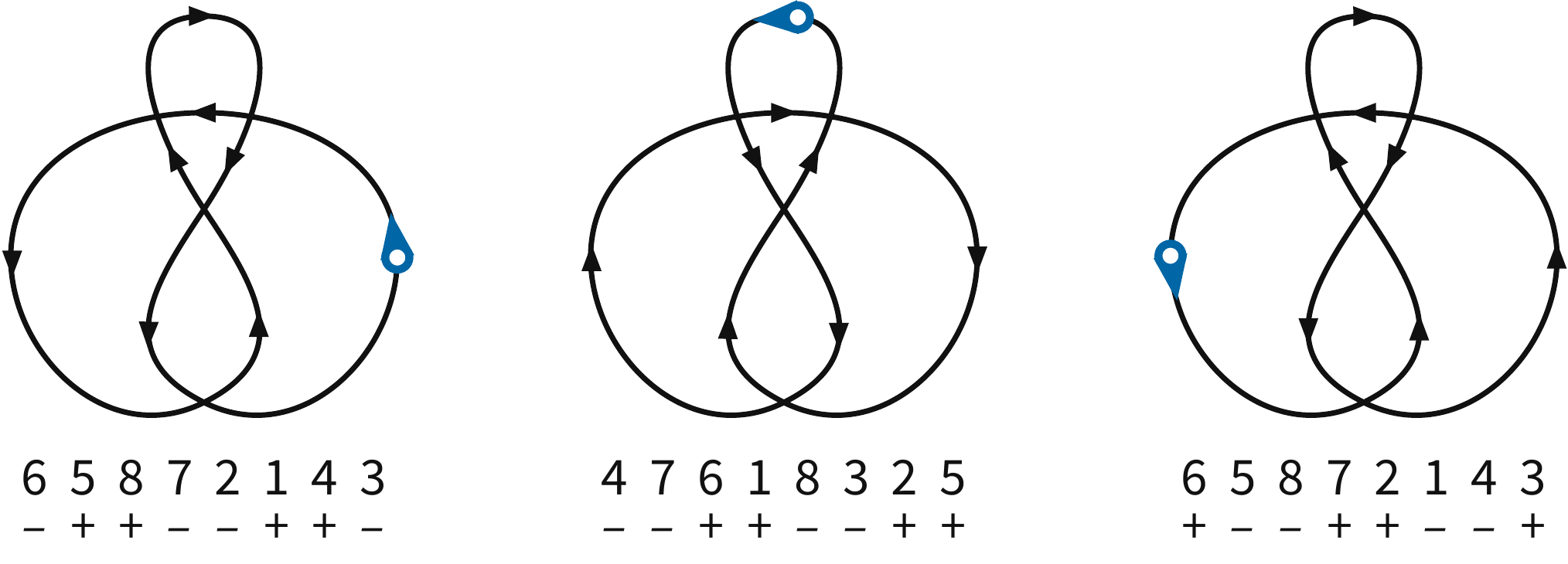}
\caption{A closed curve with exactly three valid signed crossing codes; compare with Figure \ref{F:edgecodes}.}
\label{F:codes}
\end{figure}

In particular, to obtain a valid signed crossing sequence for a generic polygon $P$, it suffices to choose any vertex on the convex hull of $P$ as the basepoint vertex $p_1$, and direct the polygon so that the triple $(p_m, p_1, p_2)$ is oriented counterclockwise.  (Different convex hull vertices may yield the same crossing code, and some valid crossing codes for $P$ may require a basepoint that is not a convex hull vertex, or not a vertex at all.)

We also associate with every generic polygon $P$ an \EMPH{edge code}, which records which edges of $P$ contain each self-crossing.  The edge code is a vector $\EMPH{edge}\in[m]^{2n}$, where for each index $i$, the $i$th crossing along~$P$ occurs at the point $p_e p_{e+1} \cap p_{e'} p_{e'+1}$, where $e = \Edge_i$ and $e' = \Edge_{\Twin_i}$.  The vector $\Edge$ is (weakly) sorted; for each index $1\le i< 2n$, we have $\Edge_i \le \Edge_{i+1}$.

Figure \ref{F:edgecodes} shows a generic polygon $P$ whose convex hull has four vertices; along with the signed crossing and edge codes that result from choosing each convex hull vertex as the basepoint.  Each signed crossing code is also a valid signed crossing code for the curve $\gamma$ in Figure \ref{F:codes} (but not vice versa); it follows that $\gamma$ and $P$ are isotopic.

\begin{figure}[ht]
\centering
\includegraphics[width=\textwidth]{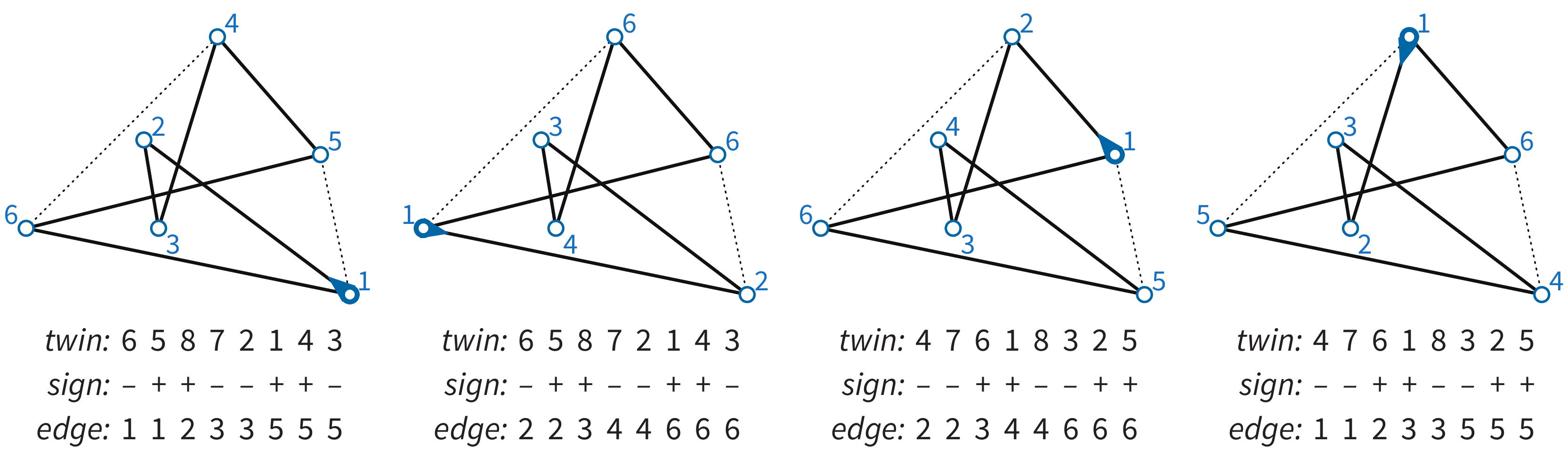}
\caption{All valid signed crossing and edge codes for the same generic polygon; compare with Figure \ref{F:codes}.}
\label{F:edgecodes}
\end{figure}

\subsection{Overview}

For the remainder of this section, we use modular index arithmetic implicitly; in particular, $p_{m+1} = p_1$.  All indexed conjunctions ($\Land$), disjunctions ($\Lor$), and summations ($\sum$) are notational shorthand; each term in these expressions appears explicitly in the actual formula.  For example, the indexed disjunction
\[
	\Land_{i=1}^{2n-1} (\Edge_{i} \le \Edge_{i+1})
\]
is notational shorthand for the explicit disjunction
\[
	(\Edge_1 \le \Edge_2) \land (\Edge_2 \le \Edge_3) \land \cdots \land (\Edge_{2n-1} \le \Edge_{2n})
\]

For any closed curve $\gamma$ and any integer $m$, we describe an ETR sentence $\textsc{IsotopicToPolygon}(\gamma, m)$ that is true if and only if some generic $m$-gon is isotopic to $\gamma$.  Lemma \ref{L:crude-upper} implies that we can assume without loss of generality that $m = O(n)$.  Our sentence has the following high-level structure:
\begin{align*}
	\textsc{IsotopicToPolygon}(\gamma, m)
	~\equiv~
	&\exists P\in\Real^{2m}\colon \exists \Edge\in\Real^{2n} \colon
	\\ & \quad
	\Lor_{(\Twin,\Sign)\sim\gamma}
	\textsc{CodedPolygon}(P, \Twin, \Sign, \Edge)
\end{align*}
Here $(\Twin,\Sign)$ ranges over all valid signed crossing codes of the input curve $\gamma$; in each term of this disjunction, the vectors $\Twin$ and $\Sign$ are hard-coded constants.  The polygon $P$ is represented by its vector $(x_1, y_1,\allowbreak x_2, y_2, \dots,\allowbreak x_m, y_m)$ of vertex coordinates; for notational convenience, we write $p_i = (x_i,y_i)$.  Crudely, the predicate $\textsc{CodedPolygon}(P, \Twin,\allowbreak \Sign, \Edge)$ states that $P$ is a generic polygon with valid signed crossing code $(\Twin, \Sign)$ and edge code $\Edge$.  Although in principle we could compute the edge code of $P$ from its vertex coordinates, we use the inherent nondeterminism in the ETR model to “guess” the edge code, which simplifies our formula.

In more detail, the predicate $\textsc{CodedPolygon}(P, \Twin, \Sign, \Edge)$ certifies the following conditions:
\begin{itemize}\itemsep0pt
\item
$P$ is a generic polygon; $p_1$ is the rightmost vertex of $P$, and the triple $(p_m,p_1,p_2)$ is oriented counterclockwise.
\item
The vertex coordinates of $P$ describe a polygon with exactly $n$ self-crossings.
\item
The vector $\Edge$ is a well-formed edge code.
\item
The crossings occur on the correct edges, with the correct signs; specifically, for each index $i$, the edges with index $\Edge_i$ and $\Edge_{\Twin_i}$ cross with sign $\Sign_i$.
\item
Finally, the crossings along each edge of $P$ occur in the correct order.
\end{itemize}
We describe subformulas for each of these conditions in the following subsections.
\begin{align*}
	\textsc{CodedPolygon}&(P, \Twin, \Sign, \Edge) ~\equiv~
	\\&  	\emph{GoodPolygon}(P) ~\land~ (\emph{NumCrossings}(P) = n) ~\land~
			\emph{WellFormed}(\Edge) ~\land~
  	\\&		\emph{CrossingSigns}(P, \Twin, \Sign, \Edge) ~\land~
 			\emph{CrossingOrder}(P, \Twin, \Edge)
\end{align*}

\subsection{Good Polygon}

To avoid boundary cases elsewhere in our formula, we require that (1) no two vertices of $P$ have the same $x$-coordinate, (2) no three vertices of $P$ are collinear, (3) no two edges of $P$ are parallel, and (4)~no three edges of $P$ lie on concurrent lines.  These conditions, which can be enforced if necessary by arbitrarily small vertex perturbations, imply that all $m$ vertices of $P$ are distinct; no edge of~$P$ is vertical; no two edges of $P$ are collinear; no vertex of $P$ lies in the interior of an edge; and at most two edges intersect at any point.  To establish a valid signed crossing code, we also require that $p_1$ is the rightmost vertex of $P$, and the triple $(p_m,p_1,p_2)$ is oriented counterclockwise.

The signed area of the triangle $(p_i, p_j, p_k)$ is given by the following determinant:
\[
	\Delta(i, j, k)
	~\equiv~
			\begin{vmatrix}
				1 & x_i & y_i \\
				1 & x_j & y_j \\
				1 & x_k & y_k
			\end{vmatrix}
	~\equiv~
		(x_i-x_k)(y_j-y_k) - (y_i-y_k)(x_j-x_k)	
\]
This signed area is positive if the triple $(p_i, p_j, p_k)$ is oriented counterclockwise, negative if the triple is oriented clockwise, and zero if the three points are collinear. Edges $p_ip_{i+1}$ and $p_j p_{j+1}$ lie on parallel lines (or one of the edges has length $0$) if and only if 
\begin{align*}
	\emph{parallel}(i,j) &~\equiv~
		\begin{vmatrix}
			x_{i+1} - x_i & y_{i+1} - y_i \\
			x_{j+1} - x_j & y_{j+1} - y_j
		\end{vmatrix}
		= 0.
\end{align*}
Similarly,  edges $p_ip_{i+1}$ and $p_j p_{j+1}$ and $p_k p_{k+1}$ lie on concurrent lines (or some pair is parallel) if and only if
\[
	\emph{concurrent}(i,j,k) ~\equiv~
		\begin{vmatrix}
			x_{i+1} - x_i & y_{i+1} - y_i & x_{i+1}y_i - x_iy_{i+1} \\
			x_{j+1} - x_j & y_{j+1} - y_j & x_{j+1}y_j - x_jy_{j+1} \\
			x_{k+1} - x_j & y_{k+1} - y_j & x_{k+1}y_j - x_ky_{k+1} \\
		\end{vmatrix}
		= 0.
\]
Our overall general position expression has length $\Theta(m^3)$:
\begin{align*}
	\emph{GoodPolygon}(P) ~\equiv~ & 
			\left(\Land_{i=1}^{m} (x_i \ne x_{i+1})\right)
			~\land~
			\left(\Land_{i=2}^{m} (x_1 \le x_i)\right)
			~\land~
			\big(\Delta(m, 1, 2) > 0\big)
	\\&		{}\qquad~\land~	\left(\Land_{i=1}^{m-1}\Land_{j=i+1}^{m}
							\lnot\emph{parallel}(i,j)\right)
	\\&		{}\qquad~\land~ \left(
				\Land_{i=1}^{m-2}\Land_{j=i+1}^{m-1}\Land_{k=j+1}^{m}
					\big((\Delta(i, j, k) \ne 0)
							~\land~ \lnot\emph{concurrent}(i,j,k)\big)\right)
\end{align*}

\subsection{Number of Crossings}

We can determine whether two edges $p_ip_{i+1}$ and $p_jp_{j+1}$ cross in their interiors by checking the orientations of all four triples of endpoints:
\begin{align*}
	\Cross(i,j) ~\equiv~ & 
		\big(\Delta(p_i,p_j,p_{j+1}) \cdot \Delta(p_{i+1},p_j,p_{j+1}) < 0\big)
	\\&
	\quad\land~
		\big(\Delta(p_i,p_{i+1},p_j) \cdot \Delta(p_i,p_{i+1},p_{j+1}) < 0\big)
\end{align*}
We count crossing edge pairs by defining new indicator variables $X_{ij}$ that record which pairs of edges cross, and then summing these indicator variables by brute force.  (These new variables are existentially quantified at the start of the top-level sentence.)
\begin{align*}
	\big(\emph{NumCrossings}(P)=n\big) ~\equiv~ &
	\left(\Land_{i=1}^{m-2}\Land_{j=i+2}^{m}
		\Big((X_{ij} = 1) \land \Cross(i,j)\big)
		~\lor~
		\Big((X_{ij} = 0) \land \lnot\Cross(i,j)\big)
 \right)
	\\&
	\qquad\land~
	\left(\Sum_{i=1}^{m-2}\Sum_{j=i+2}^{m} X_{ij} = n\right)
\end{align*}

\subsection{Well-Formed Edge Code}
The vector $\Edge$ is a well-formed edge code if and only if its is monotonically non-decreasing and each coordinate $\Edge_i$ is an integer between $1$ and $m$.
\[
	\emph{WellFormed}(\Edge) ~\equiv~
		\left(\Land_{i=1}^{2n-1} (\Edge_i \le \Edge_{i+1})\right)
		~\land~
		\left(\Land_{i=1}^{2n} ~ \Lor_{e=1}^m (\Edge_i = e)\right)
\]

\subsection{Crossing Signs}

The following refinement to our earlier crossing primitive states that the directed edge $p_i\arcto p_{i+1}$ crosses the directed edge $p_j\arcto p_{j+1}$ from right to left, defining a positive crossing:
\begin{align*}
	\Cross^+(i,j) ~\equiv~ & 
		\big(\Delta(p_i,p_j,p_{j+1}) < 0\big)
		~\land~
		\big(\Delta(p_{i+1},p_j,p_{j+1}) > 0\big)
	\\&
	\quad\land~
		\big(\Delta(p_i,p_{i+1},p_j) > 0\big)
		~\land~
		\big(\Delta(p_i,p_{i+1},p_{j+1}) < 0\big)
\end{align*}
We can indicate negative crossings similarly:
\begin{align*}
	\Cross^-(i,j) ~\equiv~ & 
		\big(\Delta(p_i,p_j,p_{j+1}) > 0\big)
		~\land~
		\big(\Delta(p_{i+1},p_j,p_{j+1}) < 0\big)
	\\&
	\quad\land~
		\big(\Delta(p_i,p_{i+1},p_j) < 0\big)
		~\land~
		\big(\Delta(p_i,p_{i+1},p_{j+1}) > 0\big)
\end{align*}
We easily verify that $\Cross^-(i,j) = \Cross^+(j,i)$.  The identity $\Cross(i,j) = \Cross^+(i,j) \lor \Cross^-(i,j)$ follows from the fact that if two segments cross, their endpoints alternate around their convex hull.
\begin{figure}[ht]
\centering\includegraphics[scale=0.5]{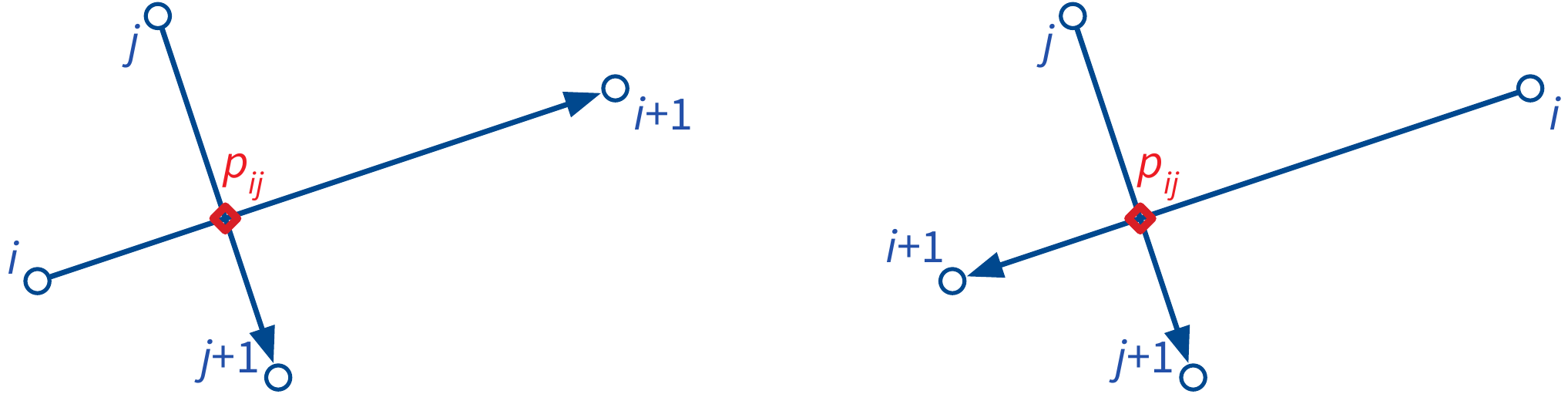}
\caption{Left: $\Cross^+(i,j) = \Cross^-(j,i)$.  ~Right: $\Cross^-(i,j) = \Cross^+(j,i)$}
\end{figure}

At this point we would like to define
\begin{align*}
	\emph{CrossingSigns}(P, \Twin, \Sign, \Edge)
	~\equiv~ 
		\Land_{i=1}^{2n}
			\bigg(&\big((\Sign_i = +1) ~\land~ \Cross^+(\Edge_i, \Edge_{\Twin_i})\big)^{\strut}
	\\&		~\lor~
			\big((\Sign_i = -1) ~\land~ \Cross^-(\Edge_i, \Edge_{\Twin_i})\big)\bigg).
\end{align*}
Unfortunately, even after expanding this expression into an explicit conjunction of $2n$ clauses, we are using the \emph{variables} $\Edge_i$ and $\Edge_{\Twin_i}$ in the $i$th clause as \emph{indices} inside the functions $\Cross^+$ and $\Cross^-$.  To simulate this indirection, we expand the offending subexpressions into disjunctions as follows:  
\begin{align*}
	\Cross^\pm(\Edge_i, \Edge_{\Twin_i})
	&~\equiv~
	\Lor_{e = 1}^{m}
	\Lor_{et = 1}^{m}
		\left((e = \Edge_i) ~\land~ (et = \Edge_{\Twin_i}) ~\land~ \Cross^\pm(e, et)\right)
\end{align*}
Again, this indexed disjunction is shorthand for an explicit disjunction of $\Theta(m^2)$ terms; within each term
\[
	\left((e = \Edge_i) \land (et = \Edge_{\Twin_i}) \land \Cross^\pm(e, et)\right)
\]
the crossing indices $i$ and $\Twin_i$ and the edge indices $e$ and $et$ are explicit constants, so all variables, including the coordinate variables in the subexpression $\Cross^+(e, et)$, are identifiable at “compile time”.

Our final expression $\emph{CrossingSigns}(P, \Twin, \Sign, \Edge)$ has length $\Theta(m^2 n)$.  As a side effect, this expression guarantees that every pair of polygon edges that is supposed to cross actually does.  Moreover, because we have already verified that the correct \emph{number} of edge-pairs cross, it follows that \emph{only} the edge pairs that are supposed to cross actually do.

\subsection{Crossing Order}

Finally, we need to ensure that the crossings along each edge of $P$ appear in the correct order.
For all indices $i\ne j$, let $p_{ij} = (x_{ij}, y_{ij})$ denote the intersection point of the lines $\Line{p_i p_{i+1}}$ and $\Line{p_j p_{j+1}}$; our general position assumptions imply that these points are well-defined.  Of course we could compute these points explicitly, but it is simpler to define them implicitly by the clauses
\[
	(\Delta(p_{ij}, p_i, p_{i+1}) = 0) ~\land~ (\Delta(p_{ij}, p_j, p_{j+1}) = 0).
\]

For any three indices $i$, $j$, and $k$, the following expression states that the endpoint $p_i$ and the intersection points $p_{ij}$ and $p_{ik}$ appear in order along the line $\smash{\Line{p_i p_{i+1}}}$:
\begin{align*}
	\emph{OrderedP}(i,j,k)
	~\equiv~
		\begin{cases}
			\textsc{False} & \text{if $i=j$ or $i=k$} \\
			(x_i - x_{ij})(x_{ij} - x_{ik}) > 0 & \text{otherwise}
		\end{cases}
\end{align*}
(Recall that our general position assumption implies that $x_i \ne x_{i+1}$.)

\begin{figure}[ht]
\centering
\includegraphics[scale=0.5]{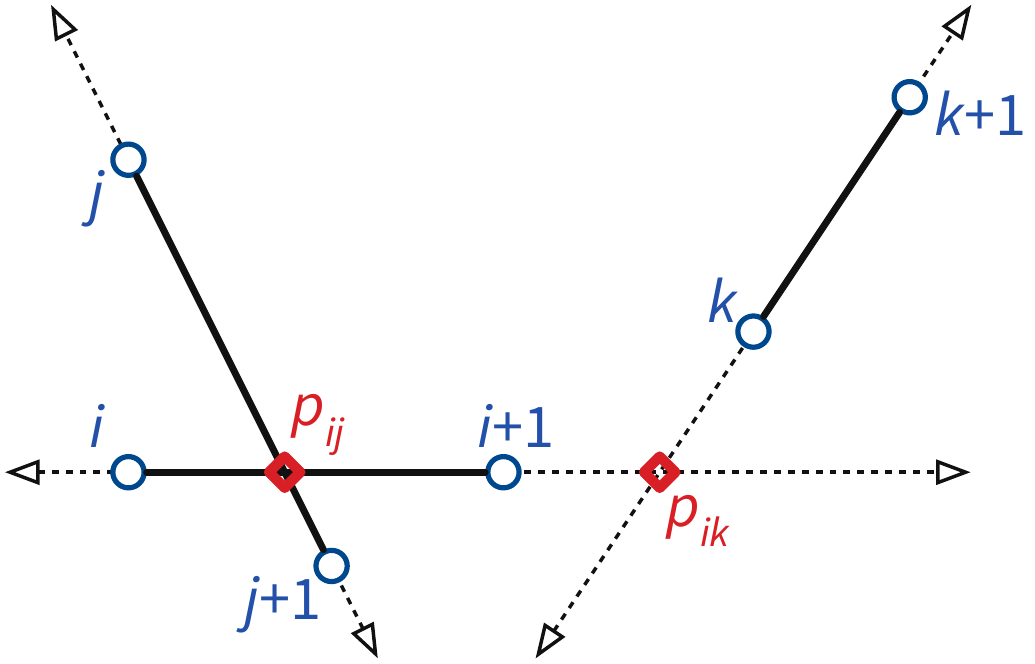}
\caption{$\emph{OrderedP}(i,j,k)$} 
\end{figure}

We need to express this relation in terms of \emph{crossing} indices, instead of \emph{edge} indices.  The following expression states that the $i$th and $(i+1)$st crossings are correctly ordered, unless they appear on different edges.
\[
	\emph{OrderedX}(i)
	~\equiv~
		(\Edge_i \ne \Edge_{i+1}) ~\lor~
			\emph{OrderedP}(\Edge_i, \Edge_{\Twin_i}, \Edge_{\Twin_{i+1}})
\]
Our earlier crossing-sign expressions guarantee that the relevant \emph{edges} actually intersect, and therefore intersect in the correct order.  As before, we expand this expression to simulate indirection, as follows.
\begin{multline*}
	\emph{OrderedP}(\Edge_i, \Edge_{\Twin_i}, \Edge_{\Twin_{i+1}})
\\	
	\equiv~ 
		\Lor_{e=1}^m ~ \Lor_{et=1}^m ~ \Lor_{et'=1}^m
			 \big( (e = \Edge_i) 
					~\land~ (et = \Edge_{\Twin_i}) 
					~\land~ (et' = \Edge_{\Twin_{i+1}})
					~\land~ \emph{OrderedP}(e, et, et')\big)
\end{multline*}
Our final ordering expression 
\[
	\emph{CrossingOrder}(P, \Twin, \Edge) ~\equiv~
		\Lor_{i=1}^{2n-1} \emph{OrderedX}(i)
\]
has length $\Theta(nm^3)$.

\begin{theorem}
For any generic curve $\gamma$ with $n$ self-intersection points and any integer $m = O(n)$, there is a sentence of length $\Theta(nm^3) = O(n^4)$ in the existential theory of the reals that is true if and only if $\gamma$ is isotopic to a generic polygon with $m$ vertices.
\end{theorem}

\begin{corollary}
\textsc{CurveToPolygon} is $\exists\Real$-complete.
\end{corollary}


\bibliographystyle{newuser-doi}
\bibliography{topology,compgeom,optimization,algorithms,lowerbounds,jeffe}

\end{document}